\documentclass[12pt]{article}

\usepackage{amsmath,amsfonts,amssymb,amsthm,textcomp,mathrsfs,mathrsfs,bbm}
\usepackage{braket,marvosym}
\usepackage{mathtools}
\usepackage[utf8]{inputenc}
\usepackage{graphicx}
\usepackage{color,xcolor}
\usepackage[colorlinks=true,linkcolor=black,citecolor=black,plainpages=false,pdfpagelabels]{hyperref}
\usepackage{palatino}
\usepackage[T1]{fontenc}
\usepackage{fullpage}

\newtheorem{theorem}{Theorem}

\newtheorem{definition}{Definition}

\newcommand*{\eye}{{\mathbbm{1}}}

\newcommand*{\cE}{\mathcal{E}}
\newcommand*{\cT}{{\mathcal{T}}}

\newcommand*{\cR}{\mathcal{R}}
\newcommand{\beq}{\begin{equation}}
\newcommand{\enq}{\end{equation}}

\newcommand{\tr}{\mathrm{Tr}}

\newcommand*{\cV}{\mathcal{V}}
\newcommand*{\cU}{\mathcal{U}}

\newcommand*{\renyi}{R\'{e}nyi }

\newcommand{\mapone}{\cT_W}

\newcommand{\err}{{\text{error}}}

\usepackage[colorlinks=true,linkcolor=black,citecolor=black,plainpages=false,pdfpagelabels]{hyperref}
\hypersetup{pdftitle={Template}}
\usepackage{color}

\allowdisplaybreaks[2]

\begin{document}

\newcommand{\itwomax}{{^2I}_{\max}}

\title{Comment on "Conditional Decoupling of Quantum Information"}
\author{Naresh Sharma\let\thefootnote\relax\footnote{Email: {\tt n.kumar.sharma@gmail.com}} \\
Pathankot 145001 \\
Punjab, India}
\date{\today}
\maketitle

\begin{abstract}
Berta \emph{et al}
[\href{http://dx.doi.org/10.1103/physrevlett.121.040504}{{\color{blue} Phys. Rev.
Lett., {\bf 121}, 040504 (2018)}}]
claim that their result provides a conceptually
new extension of the decoupling approach to quantum
information theory. We provide an alternate proof
using the plain-vanilla decoupling approach for the
achievable rates of their main result and hence, their
claim is unwarranted, and the title can be misleading
when taken in conjunction with the claim.
\end{abstract}

Berta \emph{et al} \cite{cond-corr-2018}
analyze two protocols namely
deconstruction and conditional erasure, differing
only in error constraints, closely
related to the one in \cite{corr-pra-2005}, for
which the decoupling approach works
\cite{sharma-decoupling-2015}, and make the
following claims:
\begin{itemize}
\item \emph{Our models for deconstruction
and conditional erasure extend the decoupling
approach to quantum information theory.}
\item \emph{Our result can alternatively be read as a conditional
decoupling theorem and hence provides a conceptually
new extension of the decoupling approach to quantum
information theory.}
\end{itemize}
These claims are made without qualifications and one 
could at least weakly interpret them as saying that 
their result is not merely to show that an
approach powered by the Quantum State Redistribution
(QSR) \cite{devetak-state-redist-2008, qsr-pra-2008} 
works for their protocols, moderately interpret
them as saying that their result is
beyond what a mere decoupling approach can provide,
and strongly interpret them as saying that there
is something fundamental about "conditional
decoupling". The strong interpretation looks quite 
plausible considering
the audacity and scope of the claims when taken
along with the title.

Clearly, the possibility of a
"conceptually new extension" of a method
by considering two example protocols arises only if
the method cannot address the protocols. There is no
proof provided in \cite{cond-corr-2018} that it is
impossible to analyze their protocols by using the
decoupling approach. Furthermore, the exploitation
of the QSR in \cite{cond-corr-2018} for their
achievable rates is by itself a ratification of
the decoupling approach since it can be used for QSR
as shown in \cite{sharma-decoupling-2015} rendering
the claims
inscrutable since \cite{sharma-decoupling-2015}
pre-dates \cite{cond-corr-2018}.

The authors of \cite{cond-corr-2018} harp
on the subtle differences between the two protocols.
Notice these sentences from \cite{cond-corr-2018}:
\begin{itemize}
\item[] \emph{We would like to emphasize again that 
deconstruction
and conditional erasure protocols are more delicate than
standard decoupling, the latter sometimes described as
having the relatively indiscriminate goal of destruction
[38]. That is, a straightforward application of the
decoupling method is too blunt of a tool to apply in a
state deconstruction protocol.}
\end{itemize}
Ultimately, the true test of any claims or insights is 
the proof itself, and in its absence, one cannot
fathom where this lack of faith in decoupling and a
leap of faith in something beyond decoupling are
coming from.

The definition of deconstruction requires the notion
of the optimal permissible recovery map as
solution to the problem: for a tripartite quantum
state $\rho^{ABR}$, find a map $\cR^{B \to AB}$ such that
$\cR(\rho^{BR})$ is closest to $\rho^{ABR}$.

Even if the benefit of doubt is given to the authors
in the second sentence above on the
deconstruction protocol that the recovery map one 
gets by leveraging the conditional erasure protocol
is trivial giving rise to a tensor product state,
the fact remains that this trivial map forms the basis
in \cite{cond-corr-2018} to show the achievable
rates only for the conditional erasure and the
converse takes care of the fact that it is also
optimal for the deconstruction. If the authors had
shown that the optimal rates for the deconstruction
and the
conditional erasure are different perhaps due to
certain vagaries of the optimal recovery map or,
even less, that
the trivial map one gets by the leveraging the
conditional erasure was insufficient for providing the 
optimal rate for the deconstruction and this can only be
rectified by something beyond the "blunt" decoupling,
then the second sentence may be valid. (Clearly, for
the vagaries of the optimal recovery maps to kick in,
the problem statement would have to change, but we
digress.) No such
thing was proved and a \emph{slight}
issue here is that the aforesaid optimal
rates being exactly the same is the main result of
\cite{cond-corr-2018}!

In general, mankind knows little to nothing
about the optimal 
recovery maps and in this realm of ungloomy ignorance 
where we somehow survive, to make statements that 
seemingly pertain to the other unknown realm
to pulverize a method when the proofs, including
those in \cite{cond-corr-2018}, elucidate no such
possibility is verbal voodoo, quantum or not!

We provide the achievable rates for the protocols by 
the good-old decoupling approach.
This implies that their aforesaid claims, no matter which
of the weak to strong interpretations one clings to,
and statements are uncalled-for.

Perhaps there
is an exciting new concept of "conditional decoupling" 
out there waiting to be discovered, but
\cite{cond-corr-2018} fails to deliver!

The proof is similar to the one in Section 11
in \cite{sharma-decoupling-2015}. $|A|$ is the
dimension of $A$, $\pi^A$ is the maximally mixed
state on $A$, for $\rho^{AR}$,
$H_{\alpha}(A|R)_\rho$ is the \renyi
quantum conditional entropy of $A$ given $R$
and since there are several options, one could pick one's
favorite --- see \cite{sharma-decoupling-2015}, the
von Neumann entropies for $\alpha = 1$ are denoted
by $H(A|R)_\rho$, for $\rho^{ABR}$, the conditional 
mutual information between $A$ and $R$ given $B$ is
given by $I(A;R|B)_\rho = H(A|B)_\rho - H(A|RB)_\rho$,
$\|\cdot \|_1$ is the trace norm, and
$\Xi(\varepsilon) \equiv \sqrt{\varepsilon (2 + \varepsilon + 2 \sqrt{1+
\varepsilon})}$ for $\varepsilon \geqslant 0$. For
$\Psi^{ABRE}$
pure, there are duality relations such that
$H_\alpha(A|RE)_\Psi = -H_{\widetilde{\alpha}}(A|B)_\Psi$
and the relationships between $\alpha$ and
$\widetilde{\alpha}$ are provided in
\cite{tomamichel-duality-2013} and references therein.
Define
\beq
\mapone^{A \to B}(\sigma^A) \equiv \frac{|A|}{|B|} (W^{A \to B} \cdot \sigma^A),
\enq
where $W^{A \to B}$, $|A| \geqslant |B|$, is a full-rank partial isometry.

\begin{definition}[Deconstruction and
Conditional erasure]
Consider $n$ copies of a tripartite state $\rho^{ABR}$, with $M$ Unitaries $U_i$, $i=1,...,M$, applied over
$A ^nB^n$ such that
\beq
\Upsilon^{A^nB^nR^n} \equiv \frac{1}{M} \sum_{i=1}^M U_i \cdot \left[
(\rho^{ABR})^{\otimes n} \right],
\enq
A $(\rho,\err,n)$ deconstruction protocol ensures
\beq
\sup_{\cR^{B^n \to A^n B^n}} 
\Big\| \Upsilon^{A^n B^nR^n} -
\cR^{B^n \to A^n B^n}(\Upsilon^{B^nR^n})
\Big\|_1 \leqslant \err ~~ \mbox{and} ~~
\Big\| \Upsilon^{B^nR^n} - (\rho^{BR})^{\otimes n}
\Big\|_1 \leqslant \err,
\enq
where the supremum is taken over all the completely
positive and trace preserving (cptp)
recovery maps $\cR^{B^n \to A^n B^n}$.

A$(\rho,\err,n)$ conditional erasure protocol
ensures 
\beq
\Big\| \Upsilon^{A^nB^nR^n} - \tau^{A^n} \otimes \Upsilon^{B^nR^n}
\Big\|_1 \leqslant \err ~~ \mbox{and} ~~
\Big\| \Upsilon^{B^nR^n} - (\rho^{BR})^{\otimes n}
\Big\|_1 \leqslant \err,
\enq
where $\tau^{A^n}$ is a density matrix and we make
no apriori restrictions on the choice of $\tau^{A^n}$.

The number $(\log M)/n$ is called the {\bf rate} of the protocol.
A real number ${\mathbf{r}}$ is called an {\bf achievable rate} if protocols exist for $n \to \infty$ with
rate approaching
${\mathbf{r}}$ and the $\err$ approaching $0$.
\end{definition}

\begin{theorem}[Berta \emph{et al}, 2018 \cite{cond-corr-2018}]
The smallest achievable rate for both the
protocols is $I(A:R|B)_\rho$.
\end{theorem}

We prove the following theorem.

\begin{theorem}
For any $n \in \mathbb{N}$, there exist
$(\rho,\err,n)$ deconstruction and conditional
erasure protocols such that
for any $\delta > 0$, $\alpha \in (1,2]$ and $\ket{\Psi}^{ABRE}$ a purification of $\rho^{ABR}$,
\beq
\frac{\log M}{n} = H_{\widetilde{\alpha}}(A|B)_\rho -
H_{\alpha}(A | BR)_\rho +
(|E|+|B|)|R| \frac{\log(n+1)}{n} + \delta,
\enq
and the error approaches $0$ exponentially in $n$.
\end{theorem}
\begin{proof}
Consider a partial isometry $W^{A^n \to F}$,
$|F| \leqslant |A^n|$. For $M \leqslant |F|^2$, choose
$M$ \nolinebreak{Heisenberg-Weyl}
Unitaries $V_i^F$, and
let $\cV_M^{F \to F}$ be a cptp map given by
\beq
\cV_M(\sigma^F) \equiv \frac{1}{M} \sum_{i=1}^M V_i^F \cdot \sigma^F.
\enq
Then, from Corollary 2 in \cite{sharma-decoupling-2015}, 
for any $\alpha \in (1,2]$, there exists a Unitary $U$
over $A^n$ such that
\begin{multline}
\label{yae51}
\left\| \tr_F \circ \mapone \left[ U \cdot (\rho^{ARE})^{\otimes n} \right] - (\rho^{RE})^{\otimes n} \right\|_1
\\ \leqslant 8 \exp\Big\{ \frac{\alpha-1}{2\alpha} \big[ |R| |E| \log(n+1) - n H_{\alpha}(A|RE)_{\rho} - \log|F| \big] \Big\} \\
= 8 \exp\Big\{ \frac{\alpha-1}{2\alpha} \big[ |R| |E| \log(n+1) + n H_{\widetilde{\alpha}}(A|B)_{\rho} - \log|F| \big] \Big\}
\equiv \varepsilon_n,
\end{multline}
and
\begin{multline}
\label{yae52}
\left\| \cV_M \circ \mapone \left[ U \cdot (\rho^{ABR})^{\otimes n} \right] - \pi^F \otimes (\rho^{BR})^{\otimes n}
\right\|_1 \\
\leqslant 8 \exp\Big\{ \frac{\alpha-1}{2\alpha} \big[
|B| |R| \log(n+1) -
n H_\alpha(A|BR)_{\rho} - \log M + \log|F| \big] \Big\} \equiv \vartheta_n,
\end{multline}
where in \eqref{yae52}, we have also used Lemma 23 in
\cite{sharma-decoupling-2015}.
From \eqref{yae51} and Lemma 31 in
\cite{sharma-decoupling-2015}, we claim that there
exists a Unitary $V_U^{A^nB^n}$ over $A^nB^n$ such that
\beq
\label{yae53}
\left\| W^\dag \cdot \mapone \left[ U \cdot (\Psi^{ABRE})^{\otimes n} \right] -
V_U \cdot (\Psi^{ABRE})^{\otimes n} \right\|_1
\leqslant \Xi(\varepsilon_n).
\enq
Consider now the following Unitaries over $A^n$ constructed from $V_i^F$ as
$V_i^{A^n} = W^\dag \cdot V_i^F + (\eye^A - W^\dag W)$. Note that $V_i^{A^n} W^\dag = W^\dag V_i^F$.
We now claim that $V_i^{A^n} V_U$, $i=1,...,M$,
are precisely the $M$ Unitaries we need. For
\beq
\Upsilon^{A^n B^n R^n} \equiv
\frac{1}{M} \sum_{i=1}^M (V_i^{A^n} V_U) \cdot (\rho^{ABR})^{\otimes n},
\enq
we have
\begin{multline}
\label{yae1}
\Big\| \Upsilon^{A^n B^n R^n}
- (W^\dag \cdot \pi^F) \otimes (\rho^{BR})^{\otimes n} \Big\|_1 \\
\leqslant \left\| \frac{1}{M} \sum_{i=1}^M (V_i^{A^n} V_U) \cdot (\rho^{ABR})^{\otimes n}
- \frac{1}{M} \sum_{i=1}^M (V_i^{A^n} W^\dag) \cdot \mapone \left[ U \cdot (\rho^{ABR})^{\otimes n} \right] \right\|_1 + \\
\hspace{0.4in} \left\| \frac{1}{M} \sum_{i=1}^M (V_i^{A^n}
W^\dag) \cdot \mapone \left[ U \cdot (\rho^{ABR})^{\otimes n} 
\right] -
(W^\dag \cdot \pi^F) \otimes (\rho^{BR})^{\otimes n} \right\|_1 \\
\leqslant \frac{1}{M} \sum_{i=1}^M \Big\|  (V_i^{A^n} V_U) \cdot (\rho^{ABR})^{\otimes n}
- (V_i^{A^n} W^\dag) \cdot \mapone \left[ U \cdot 
(\rho^{ABR})^{\otimes n} \right] \Big\|_1 + \\
\hspace{0.83in}
\left\| \frac{1}{M} \sum_{i=1}^M (W^\dag V_i^B) \cdot \mapone \left[ U \cdot (\rho^{ABR})^{\otimes n} \right]
- (W^\dag \cdot \pi^F) \otimes (\rho^{BR})^{\otimes n} \right\|_1 \\
\leqslant \frac{1}{M} \sum_{i=1}^M \Big\|  V_U \cdot (\Psi^{ABRE})^{\otimes n}
- W^\dag \cdot \mapone \left[ U \cdot (\Psi^{ABRE})^{\otimes n} \right] \Big\|_1
+ \\
\hspace{0.43in}
\left\| \cV_M \circ \mapone \left[ U \cdot (\rho^{ABR})^{\otimes n} \right]
- \pi^F \otimes (\rho^{BR})^{\otimes n} \right\|_1 \\
\leqslant \Xi(\varepsilon_n) + \vartheta_n,
\hspace{3.8in}
\end{multline}
where the first inequality follows from the triangle inequality, in the second inequality, the
first term follows from the convexity of the trace norm and the second term follows by invoking
$V_i^{A^n} W^\dag = W^\dag V_i^B$, in the third inequality, the first term follows by invoking the
Unitary invariance and monotonicity of the trace norm
and the second term from monotonicity,
in the fourth inequality, the first term is upper bounded
using \eqref{yae53} and the second term is upper bounded using \eqref{yae52}.

Using monotonicity and \eqref{yae1}, we get
\beq
\label{yae2}
\Big\| \Upsilon^{B^n R^n}
- (\rho^{BR})^{\otimes n} \Big\|_1 \leqslant 
\Xi(\varepsilon_n) + \vartheta_n,
\enq
and using \eqref{yae1}, \eqref{yae2},
and triangle inequality, we get
\beq
\label{yae3}
\Big\| \Upsilon^{A^n B^n R^n} - (W^\dag \cdot \pi^F) 
\otimes \Upsilon^{B^n R^n} \Big\|_1 \leqslant
2 \, \Xi(\varepsilon_n) + 2 \vartheta_n, \\
\enq
which along with \eqref{yae2} proves the claim for the 
conditional erasure protocol. By substituting the
recovery map (given in \cite{cond-corr-2018}) in
\eqref{yae3} as
\beq
\cR^{B^n \to A^n B^n}(\Upsilon^{B^nR^n})
= W^\dag \cdot \pi^F \otimes \Upsilon^{B^nR^n},
\enq
and from \eqref{yae2}, the claim also follows
for the deconstruction protocol. 
\end{proof}

\noindent {\bf Remarks}:
\begin{itemize}
\item Unlike the proof in \cite{cond-corr-2018}, we
do not need any ancilla.
\item Unsurprisingly, with similar approach as
\cite{qcmi-lb-2-2015,for-mult-2016} that use QSR,
one could instead derive lower bounds to $I(A;R|E)_\rho$
using \eqref{yae1} since
for $V_{\cU}^{A^nB^n \to A^nB^nX}$ as the Stinespring
dilation isometry mocking up the application of $M$
Unitaries over $A^n B^n$ with $|X| = M$,
\eqref{yae1} implies the existence of
a recovery operation $\cE^{XE^n \to \widetilde{A}^nE^n}$
such that
\beq
\cE^{XE^n \to \widetilde{A}^nE^n} \circ \tr_{A^nB^n}
\circ V_{\cU}^{A^nB^n \to A^nB^nX}
[(\Psi^{ABRE})^{\otimes n}] \approx
\rho^{\widetilde{A}^nR^nE^n}
\enq
and noting that
\beq
\tr_{A^nB^nX}
\circ V_{\cU}^{A^nB^n \to A^nB^nX}
[(\Psi^{ABRE})^{\otimes n}] = (\rho^{RE})^{\otimes n}.
\enq
\item The converse can be proved by standard arguments
such as continuity and hence, contains no
surprises that go against this comment.
\end{itemize}


\begin{thebibliography}{1}
\providecommand{\url}[1]{\texttt{#1}}
\providecommand{\urlprefix}{URL }

\bibitem{cond-corr-2018}
{M. Berta}, {F. G. S. L. Brand\~ao}, {C. Majenz}, and {M. M. Wilde}.
\newblock
  \href{http://dx.doi.org/10.1103/physrevlett.121.040504}{\emph{Conditional
  decoupling}}\href{http://dx.doi.org/10.1103/physrevlett.121.040504}{\emph{ of quantum information}}.
  \href{https://arxiv.org/abs/1808.00135}{{\color{blue}arXiv: 1808.00135}}.
\newblock \emph{Phys. Rev. Lett.}, vol. 121: p. 040504, Jul. 2018.

\bibitem{corr-pra-2005}
B.~Groisman, S.~Popescu, and A.~Winter.
\newblock
  \href{https://link.aps.org/doi/10.1103/PhysRevA.72.032317}{\emph{Quantum,
  classical, and total amount of }}\linebreak\href{https://link.aps.org/doi/10.1103/PhysRevA.72.032317}{\emph{correlations in a quantum state}}.
\newblock \emph{Phys. Rev. A}, vol.~72: p. 032317, Sep. 2005.

\bibitem{sharma-decoupling-2015}
{N. Sharma}.
\newblock \href{https://arxiv.org/abs/1504.07075}{\emph{Random coding exponents
  galore via decoupling}}.
  \href{https://arxiv.org/abs/1504.07075}{{\color{blue}arXiv: 1504.07075}}.

\bibitem{devetak-state-redist-2008}
{I. Devetak} and {J. Yard}.
\newblock \href{http://dx.doi.org/10.1103/PhysRevLett.100.230501}{\emph{Exact
  cost of redistributing multipartite quantum states}}.
\newblock \emph{Phys. Rev. Lett.}, vol. 100: p. 230501, Jun. 2008.

\bibitem{qsr-pra-2008}
{M.-Y. Ye}, {Y.-K. Bai}, and {Z. D. Wang}.
\newblock \href{http://dx.doi.org/10.1103/PhysRevA.78.030302}{\emph{Quantum
  state redistribution based on a generalized}}
  \href{http://dx.doi.org/10.1103/PhysRevA.78.030302}{\emph{decoupling}}.
  \href{https://arxiv.org/abs/0805.1542}{{\color{blue}{arXiv: 0805.1542}}}.
\newblock \emph{Phys. Rev. A}, vol.~78: p. 030302, Sep. 2008.

\bibitem{tomamichel-duality-2013}
{M. Tomamichel}, {M. Berta}, and {M. Hayashi}.
\newblock \href{http://dx.doi.org/10.1063/1.4892761}{\emph{Relating different
  quantum generalizations of}} \href{http://dx.doi.org/10.1063/1.4892761}{\emph{the conditional {R}\'{e}nyi entropy}}.
  \href{https://arxiv:1311.3887}{{\color{blue}arXiv:1311.3887}}.
\newblock \emph{J. Math. Phys.}, vol.~55: p. 082206, 2014.

\bibitem{qcmi-lb-2-2015}
F.~G. S.~L. Brand\~ao, A.~W. Harrow, J.~Oppenheim, and S.~Strelchuk.
\newblock
  \href{https://link.aps.org/doi/10.1103/PhysRevLett.115.050501}{\emph{Quantum
  conditional}}\href{https://link.aps.org/doi/10.1103/PhysRevLett.115.050501}{\emph{ mutual information, reconstructed states, and state
  redistribution}}. \href{https://arxiv.org/abs/1411.4921}{{\color{blue}arXiv:
  1411.4921}}.
\newblock \emph{Phys. Rev. Lett.}, vol. 115: p. 050501, Jul. 2015.

\bibitem{for-mult-2016}
M.~{Berta} and M.~{Tomamichel}.
\newblock \href{http://dx.doi.org/10.1109/TIT.2016.2527683}{\emph{The fidelity
  of recovery is multiplicative}}.
  \href{https://arxiv.org/abs/1502.07973}{{\color{blue}arXiv: 1502.07973}}.
\newblock \emph{IEEE Trans. Inf. Theory}, vol.~62: pp. 1758--1763, Apr. 2016.

\end{thebibliography}
\end{document}